\documentclass[conference]{IEEEtran}
\usepackage{amsthm, amssymb, amsmath}
\usepackage{verbatim, float}
\usepackage{mathrsfs}
\usepackage{graphics}
\usepackage[usenames,dvipsnames]{pstricks}
\usepackage{epsfig}
\usepackage{pst-grad} % For gradients
\usepackage{pst-plot} % For axes
\usepackage{cases}
\usepackage{algorithmic}
\usepackage{bm}
\usepackage[noadjust]{cite}
\usepackage[top=1.4cm, bottom=0.52in, left=0.6in, right=0.6in]{geometry}
\usepackage{url}

\usepackage{hhline}

\usepackage{etoolbox}
\makeatletter
\patchcmd{\@makecaption}
  {\scshape}
  {}
  {}
  {}
\makeatother
\usepackage[singlelinecheck=false]{caption}
\usepackage{diagbox}

\theoremstyle{plain}

\newtheorem{proposition}{Proposition}

\theoremstyle{definition}
\newtheorem{definition}{Definition}

\newcommand{\C}{{\mathcal C}}

\newcommand{\R}{{\mathcal R}}

\newcommand{\bI}{{\boldsymbol I}}

\newcommand{\bH}{{\boldsymbol H}}

\newcommand{\bC}{{\boldsymbol C}}

\newcommand{\bV}{\boldsymbol{V}}

\newcommand{\bb}{{\boldsymbol b}}
\newcommand{\bu}{{\boldsymbol u}}

\newcommand{\by}{{\boldsymbol y}}

\newcommand{\bx}{{\boldsymbol{x}}}
\newcommand{\bz}{{\boldsymbol{z}}}

\newcommand{\bG}{{\boldsymbol{G}}}

\newcommand{\bbj}{\boldsymbol{\beta}_j}

\newcommand{\bgaj}{\boldsymbol{\gamma}_j}

\newcommand{\bc}{{\boldsymbol c}}

\newcommand{\bcj}{{\boldsymbol c}_j}
\newcommand{\bci}{{\boldsymbol c}_i}
\newcommand{\bcjs}{{\boldsymbol c}_{j^*}}
\newcommand{\bco}{{\boldsymbol c}_1}
\newcommand{\bcn}{{\boldsymbol c}_n}
\newcommand{\bct}{{\boldsymbol c}_2}

\newcommand{\bgi}{{\vec{\boldsymbol{g}}^{(i)}}}

\newcommand{\bgij}{{\vec{\boldsymbol{g}}^{(i)}_j}}
\newcommand{\bgijs}{{\vec{\boldsymbol{g}}^{(i)}_{j^*}}}

\newcommand{\fte}{{\mathbb F}_{256}}

\newcommand{\Ff}{{\mathbb F}_4}
\newcommand{\fst}{{\mathbb F}_{16}}

\newcommand{\ff}{\mathbb{F}}
\newcommand{\ft}{\mathbb{F}_2}
\newcommand{\fq}{\mathbb{F}_q}

\newcommand{\fql}{\mathbb{F}_{q^{\ell}}}

\newcommand{\fqm}{\mathbb{F}_{q^m}}

\newcommand{\bal}{\bm{\alpha}}

\newcommand{\balb}{\overline{\bm{\alpha}}}

\newcommand{\baljs}{\bm{\alpha}_{j^*}}
\newcommand{\balj}{\bm{\alpha}_j}

\newcommand{\bbe}{\bm{\beta}}

\newcommand{\bom}{\bm{\omega}}

\newcommand{\bga}{\bm{\gamma}}
\newcommand{\bla}{\bm{\lambda}}

\newcommand{\spn}{\mathsf{span}_{\fq}}

\newcommand{\rank}{\mathsf{rank}_{\fq}}

\newcommand{\srt}{{\sf{RepairTrace}}}
\newcommand{\sct}{{\sf{ColumnTrace}}}

\renewcommand{\sp}{{\sf{Parity}}}

\newcommand{\vbc}{\vec{\boldsymbol{c}}}
\newcommand{\vc}{\vec{\boldsymbol{c}}}
\newcommand{\vu}{\vec{\boldsymbol{u}}}

\newcommand{\vx}{\vec{\boldsymbol{x}}}
\newcommand{\vy}{\vec{\boldsymbol{y}}}
\newcommand{\vlam}{\vec{\boldsymbol{\lambda}}}
\newcommand{\vga}{\vec{\boldsymbol{\gamma}}}

\newcommand{\dq}{\dim_{\fq}}

\newcommand{\vb}{\vec{b}}
\newcommand{\vbs}{\vec{b}^*}
\newcommand{\mvb}{\max\big(\vec{b}\big)}

\newcommand{\lgr}{\left\{g_i(x)\right\}_{i\in[\ell]}}
\newcommand{\lhr}{\left\{h_i(x)\right\}_{i\in[\ell]}}
\newcommand{\Rg}{\R\left(\{\bgi\}_{i\in [\ell]}\right)}

\newcommand{\rvline}{\hspace*{-\arraycolsep}\vline\hspace*{-\arraycolsep}}

\newcommand{\Cd}{\mathcal{C}^\perp}
\newcommand{\rsk}{\text{RS}(A,k)}
\newcommand{\grskl}{\text{GRS}(A,k,\vec{\boldsymbol{\lambda}})}
\newcommand{\grsnkl}{\text{GRS}(A,n-k,\vec{\boldsymbol{\lambda}})}

\newcommand{\faj}{f(\boldsymbol{\alpha}_j)}

\newcommand{\gaj}{g(\boldsymbol{\alpha}_j)}

\begin{document}

\title{Practical Considerations in\\ Repairing Reed-Solomon Codes
\thanks{
T. X. Dinh, L. J. Mohan, and S. H. Dau are with the School of Computing Technologies, RMIT University, Australia. Emails: \{S3880660, sonhoang.dau\}@rmit.edu.au.  
T. X. Dinh is also with the Department of Mathematics, Faculty of Natural Science and Technology, Tay Nguyen University, Vietnam.
L. Y. N. Nguyen was an intern at RMIT Univesity when this work was done.
Tran Thi Luong is with the Department of Information Security, Academy of Cryptography Techniques, Hanoi, Vietnam. Email: ttluong@bcy.gov.vn.}}

\author{Thi Xinh Dinh$^\dag$, Luu Y Nhi Nguyen$^\dag$, Lakshmi J. Mohan$^\dag$, Serdar Boztas$^\dag$, Tran Thi Luong$^\ddag$, and Son Hoang Dau$^\dag$\\$^\dag$RMIT University, $^\ddag$Academy of Cryptographic Technique, Hanoi, Vietnam}

\date{}
\maketitle
%\pagestyle{empty}
%%%%%%%%%%%%%%%%%%%%%%%%%%%%%%%%%%

\begin{abstract}
The issue of repairing Reed-Solomon codes currently employed in industry has been sporadically discussed in the literature. 
In this work we carry out a systematic study of these codes and investigate important aspects of repairing them under the trace repair framework, including which evaluation points to select and how to implement a trace repair scheme efficiently. In particular, we employ different heuristic algorithms to search for low-bandwidth repair schemes for codes of short lengths with typical redundancies and establish three tables of current best repair schemes for $[n,k]$ Reed-Solomon codes over $\text{GF}(256)$ with $4 \leq n \leq 16$ and $r = n - k \in \{2,3,4\}$. The tables cover most known codes currently used in the distributed storage industry. 
\end{abstract}%\begin{IEEEkeywords}
%Reed-Solomon code, erasure codes, distributed storage system, repair bandwidth, heuristics.
%\end{IEEEkeywords}

\section{Introduction}
\label{sec:intro}

Reed-Solomon codes~\cite{ReedSolomon1960}, invented more than 60 years ago, still constitute the most widely used family of erasure codes in distributed storage systems to date (see Table~\ref{tab:DSS}).
Due to their popularity in practice as well as their fundamental role in the development of classical coding theory, a significant amount of research has been conducted in recent years to improve their \textit{repair bandwidth} and \textit{I/O cost} required in recovering a single or multiple erasures~\cite{Shanmugam2014,GuruswamiWootters2016, GuruswamiWootters2017,YeBarg_ISIT2016,YeBarg_TIT2017,ChowdhuryVardy2017, ChowdhuryVardy_TIT_2019,DauMilenkovic2017,DuursmaDau2017,LiWangJafarkhani-Allerton-2017,LiWangJafarkhani_TIT_2019,TamoYeBarg2017,TamoYeBarg2018,
BermanBuzagloDorShanyTamo_ISIT_2021,ConTamo_arxiv_2021,
DauDuursmaKiahMilenkovicTwoErasures2017,DauDuursmaKiahMilenkovic2018,BartanWootters2017,
MardiaBartanWootters2018,ZhangZhang_ISCIT_2019,
DauDinhKiahTranMilenkovic2021, LiWangJafarkhani_CommLett_2021,DauDuursmaChu-ISIT-2018,DauViterbo-ITW-2018,LiDauWangJafarkhaniViterbo_ISIT2019}.
In the context of distributed storage systems, the repair bandwidth of an erasure code refers to the amount of data \textit{downloaded} from the helper nodes by a recovery node to reconstruct its lost content, while the I/O cost is the total amount of data \textit{read} from the local disks of the helper nodes.% in the process.  

\begin{table}[htb!]
\centering
\includegraphics[scale=0.32]{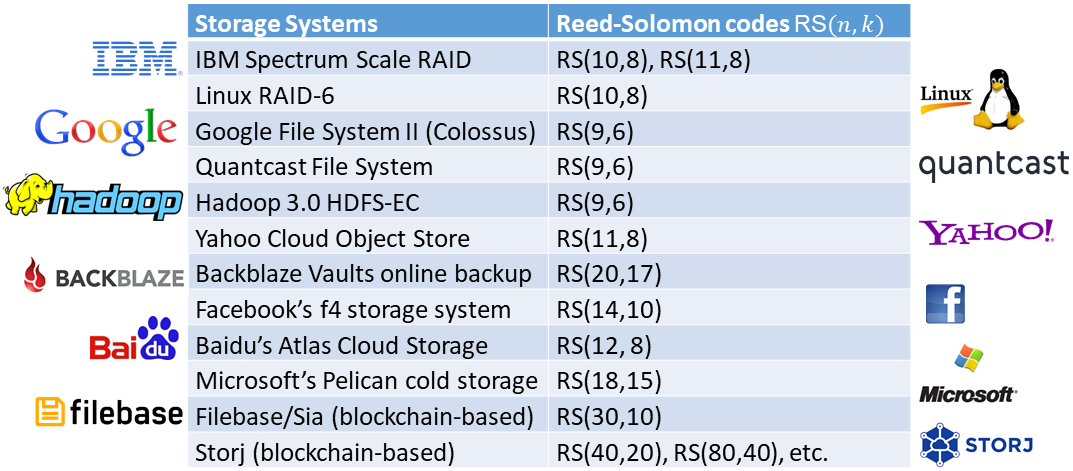}
\caption{A table of Reed-Solomon codes employed in major distributed storage systems - an updated version of~\cite[Table~I]{DauDuursmaKiahMilenkovic2018}.}
\label{tab:DSS}
\end{table}

%Apart from a few exceptions, the majority of related works in the literature can generate optimal-bandwidth repair schemes only for Reed-Solomon codes with special parameters, e.g., with length $n = q^m$ and redundancy $r = q^s$ in~\cite{GuruswamiWootters2016,GuruswamiWootters2017,DauMilenkovic2017,DauDinhKiahTranMilenkovic2021,BermanBuzagloDorShanyTamo_ISIT_2021}. Although some of these schemes still work for other $n$ and $r$,  their bandwidths may deteriorate quickly as $n$ and especially~$r$ deviate from the powers of $q$. However, most Reed-Solomon codes used in industry, as shown in Table~\ref{tab:DSS}, do not have such special parameters.  
Despite the growing literature, the treatment of short-length Reed-Solomon codes used in (or relevant to) practical storage systems has been sporadic and rather limited: RS(14,10) (used in Facebook's f4) has received the most attention~\cite{Shanmugam2014,GuruswamiWootters2016, GuruswamiWootters2017, DuursmaDau2017, LiWangJafarkhani_TIT_2019, LiWangJafarkhani_CommLett_2021}, while RS(5,3) and RS(6,4) were investigated in~\cite{Shanmugam2014}, and RS(11,8) and RS(12,8) were discussed in~\cite{LiWangJafarkhani_CommLett_2021}. Apart from \cite{DuursmaDau2017} and \cite{LiWangJafarkhani_CommLett_2021}, in which they were studied as the main topic of interest, these codes were mostly used as examples to demonstrate the inefficiency of the naive repair and the improvements in repair bandwidths that a carefully designed repair scheme can bring. 
We address this gap in the literature by providing a systematic investigation of low-bandwidth repair schemes for short-length Reed-Solomon codes that are relevant for the data storage industry and discuss several practical aspects including the selection of the evaluation points and implementation.

We first revisit four constructions of Reed-Solomon codes in existing implementations, observing that all but one are the same as the classical construction using polynomial evaluations (Section~\ref{sec:constructions}). 
Next, we discuss heuristic algorithms that can be used for construction of repair schemes for such codes and as a result, establish three tables of current-best repair schemes for codes of length $n \leq 16$ and redundancy $r \leq 4$ (Section~\ref{sec:algorithms}).
We also study in this section the impact of the evaluation points on repair bandwidths, demonstrating with an example that codes having the same $n$ and $k$ but using different lists of evaluation points may end up having different repair bandwidths.
Finally, we propose an efficient way to implement (in C) a trace repair scheme for Reed-Solomon code based on lookup tables and fast bitwise operations~\cite{implementation} on top of the state-of-the-art Intel Intelligent Storage Acceleration Library (ISA-L) (Section~\ref{sec:implementation}).  

\section{Definitions and Notations}
\label{sec:pre}

Let $[n]\triangleq\{1,2,\ldots,n\}$ and $[m,n]\triangleq\{m,m+1,\ldots,n\}$. Let~$\fq$ be the finite field of $q$ elements and $\fql$ be its extension field of degree $\ell$, where $q$ is a prime power.
In this work we only consider the case $q^\ell=256$. 
The field $\fte$ can also be viewed as a vector space over its subfields, i.e., $\fte \cong \fst^2 \cong \Ff^4 \cong \ft^8$.
Each element $\bb$ of $\fte$ can be represented as one byte, i.e., a vector of eight bits $(b_1,b_2,\ldots,b_8)$, or an integer in $[0,255]$. 
For instance, $6 = 2^2+2$ is represented by the vector $00000110$ and corresponds to $\bb =\bz^2+\bz$, where $\bz=2$ is a primitive element of $\fte$.
To accelerate the computation over $\fte$, %in practical implementations, 
additions between integers representing finite field elements are performed bitwise while multiplications are based on table lookups.
Bitwise operators in C including XOR `\^{}', AND `\&', and bit-shift `$\ll$' are heavily used in code optimization.

We use $\spn(U)$ to denote the $\fq$-subspace of $\fql$ spanned by a subset $U\subseteq\fql$. 
We use $\dq(\cdot)$ and $\rank(\cdot)$ to denote the dimension of a subspace and the rank of a set of vectors over~$\fq$. %, respectively.
% (the subscript $\fq$ is usually dropped).
The (field) trace of an element $\bb \in \fql$ over $\fq$~is 
$\mathsf{Tr}_{\fql / \fq}(\bb) \triangleq \sum_{i = 0}^{\ell-1} \bb^{q^i}$. %When clear from the context, we omit the subscripts $\fql / \fq$. 
Given an $\fq$-subspace $W$ of $\fql$, the polynomial $L_W(x)\hspace{-2pt}=\hspace{-2pt}\prod_{w\in W}(x-w)$ is called the \textit{subspace polynomial} corresponding to $W$.

A \emph{linear $[n,k]$ code} $\C$ over $\fql$ is an $\fql$-subspace of $\fql^n$ of dimension $k$. Each element $\vbc=(\bco,\bct,\ldots,\bcn)\in \C$ is referred to as a \emph{codeword} and each component $\bcj$ is called a codeword symbol. 
The \emph{dual} $\Cd$ of a code $\C$ is the orthogonal complement of $\C$ in $\fql^n$ and has dimension $r = n - k$.
The elements of $\Cd$ are called \textit{dual codewords}. %\vspace{-5pt}
We call $r$ the \emph{redundancy} of the code.
A matrix $\bG\in \fql^{k\times n}$ of rank $k$ over $\fql$ whose rows are codewords of $\C$ is called a \textit{generator matrix} of the code.
Given a generator matrix $\bG$, a message $\vu = (\bu_1,\ldots,\bu_k)$ is transformed into a codeword $\vc = \vu\bG$.
A \textit{parity check matrix} of $\C$ is simply a generator matrix $\bH$ of the dual code $\Cd$. 

\section{Existing Constructions of Reed-Solomon Codes}
\label{sec:constructions}

\subsection{Classical Construction by Reed and Solomon}
\label{subsec:classic}

The following construction of Reed-Solomon codes is the original one proposed by Reed and Solomon in
\cite{ReedSolomon1960}.

\begin{definition} 
\label{def:RS}
Let $\fql[x]$ denote the ring of polynomials over~$\fql$. A Reed-Solomon code $\rsk \subseteq \fql^n$ of dimension $k$ % over a finite field $\fql$ 
with evaluation points $A=\{\bal_j\}_{j=1}^n \subseteq \fql$
is defined as \vspace{-5pt} 
\[
\rsk = \Big\{\big(f(\bal_1),\ldots,f(\bal_n)\big) \colon f \in \fql[x],\ \deg(f) < k \Big\}. \vspace{-5pt}
\]
We also use the notation RS$(n,k)$, ignoring the evaluation points. Clearly, a generator matrix of this code is the Vandermonde matrix $\text{Vand}(\bal_1,\ldots,\bal_n) \triangleq \left(\bal_j^{i-1}\right)_{1\leq i\leq k, 1\leq j\leq n}$.
%The Reed-Solomon code is called \emph{full length} if $n=|A| = 256$. %We consider only full-length codes in this work. 
\end{definition}

A \emph{generalized} Reed-Solomon code, $\grskl$, where $\vec{\bla} = (\bla_1,\ldots,\bla_n)\in \fql^n$, is defined similarly to a Reed-Solomon code, except that the codeword
corresponding to a polynomial $f$ is defined as $\big( \bla_1f(\bal_1),\ldots,\bla_n f(\bal_n) \big)$, where $\bla_j \neq 0$ for all $j \in [n]$. 
It is well known that the dual of a Reed-Solomon code $\rsk$ is a generalized Reed-Solomon code $\grsnkl$, 
%\textcolor{red}{you do not have the arrow here}, 
for some multiplier vector $\vec{\bla}$~(\cite[Chp.~10]{MW_S}). 
We sometimes use the notation GRS$(n,k)$, ignoring $A$ and $\vlam$.

We often use $f(x)$ to denote a polynomial of degree at most $k-1$, which corresponds to a codeword of the Reed-Solomon code $\C=\rsk$, and $g(x)$ to denote a polynomial of degree at most $r-1=n-k-1$, which corresponds to a codeword of the dual code $\Cd$. Since
$
\sum_{j=1}^n \gaj(\bla_j\faj) = 0, 
$
we also refer to the polynomial $g(x)$ as a \textit{check} polynomial for $\C$. %Note that when $n = q^\ell$, we have $\bla_j = 1$ for all $j \in [n]$. 
%In general, as the column multipliers $\bla_j$ do not play any role in evaluating the repair bandwidth, they are often omitted to simplify the notation (see also Remark~\ref{rm:RS}). 
%In general, as recovering $\fa$ is equivalent to recovering $\bla_{\bal}\fa$, to simplify the notation, we omit the factor $\bla_{\bal}$ in our derivations.
%\vspace{-5pt}
 
Next, we discuss four main constructions of Reed-Solomon codes found in practical systems, three of which are equivalent to the original construction and one is invalid. All are over $\fte$.\vspace{-5pt}

\subsection{Constructions in Intelligent Storage Acceleration Library}

Both implementations of Reed-Solomon codes provided by ISA-L are systematics~\cite{ISAL}. 
The first one, also found in the Quantcast File System~\cite{QFS}, uses the generator matrix $\bG = [\bI_k \mid \bV]$ in which $\bV$ is a $k\times (n-k)$ Vandermonde matrix: $\bV = (\bx_j^{i-1})$, $i \in [1,k]$, $j \in [1,n]$, where $\bx_j = \bz^{j-1}$ and $\bz=2$ is a primitive element of $\fte$. This is \textit{not} a construction of a Reed-Solomon code, not even an MDS code (i.e., achieving the Singleton bound~\cite{MW_S}). Indeed, it is known that $\bG = [\bI_k \mid \bV]$ generates an MDS code only if every square submatrix of $\bV$ is invertible (\cite[Ch.~11, Thm.~8]{MW_S}), which is not true for a Vandermonde matrix in general. Although when $n = 9$ and $k = 6$, as in the Quantcast File System~\cite{QFSPaper2013}, the code is still MDS, we ignore this construction as it is \textit{incorrect} in general.

We focus on the second construction, which uses the generator matrix $\bG = [\bI_k \mid \bC]$ in which $\bC$ is a $k\times (n-k)$ Cauchy matrix: $\bC=\left(c_{i,j}=\frac{1}{\bx_i+\by_j}\right)$, $i \in [1,k]$, $j \in [1,n-k]$, where $\bx_i = i-1$ and $\by_j = k+j-1$. Note that $\bx_i$ and $\by_j$ are written using the integer representations of elements of $\fte$ and $\bx_i+\by_j$ refers to the (bitwise) addition of two field elements.
According to~\cite[Thm.~1]{RothSeroussi1985}, this code is the same as $\grskl$ with $A = [\vx \mid \vy] = [0,n-1]$ and $\vlam=(\bla_1,\ldots,\bla_n)$ defined as \vspace{-5pt}
%\[
%\bla_j = 
%\begin{cases}
%\frac{1}{\prod_{s = 1, s \neq j}^k (\bx_s+\bx_j)},& 1 \leq j \leq k,\\
%\frac{1}{\prod_{s = 1}^k (\bx_s+\by_{j-k})},& k+1 \leq j \leq n. 
%\end{cases}
%\]
\[
\bla_j = 
\begin{cases}
1/\prod_{s = 1, s \neq j}^k (\bx_s+\bx_j),& 1 \leq j \leq k,\\
1/\prod_{s = 1}^k (\bx_s+\by_{j-k}),& k+1 \leq j \leq n. 
\end{cases}\vspace{1pt}
\]
For example, for $n = 9, k = 6$, the Cauchy-based construction of RS(9,6) code, or more precisely, a GRS(9,6), uses $\vx = (0, 1, 2, 3, 4, 5) $ and $\vy = (6,7,8)$ and has a generator matrix\vspace{-3pt}
%\begin{footnotesize}
\[
\bG\hspace{-2pt} = \hspace{-2pt}\begin{pmatrix}
%1 & 0 & 0 & 0 & 0 & 0 & 122 & 186 & 173 \\
%0 & 1 & 0 & 0 & 0 & 0 & 186 & 122 & 157 \\
%0 & 0 & 1 & 0 & 0 & 0 & 71  & 167 & 221 \\
%0 & 0 & 0 & 1 & 0 & 0 & 167 & 71  & 152 \\
%0 & 0 & 0 & 0 & 1 & 0 & 142 & 244 & 61 \\
%0 & 0 & 0 & 0 & 0 & 1 & 244 & 142 & 170 \\
\bI_6 & \rvline & 
\begin{matrix}
122 & 186 & 173 \\
186 & 122 & 157 \\
71  & 167 & 221 \\
167 & 71  & 152 \\
142 & 244 & 61 \\
244 & 142 & 170
\end{matrix}
\end{pmatrix}
\hspace{-4pt}=\hspace{-4pt}
\begin{pmatrix}
\bI_6 & \rvline & 
  \begin{matrix}
  \bz^{229} & \bz^{57} & \bz^{252} \\
  \bz^{57} & \bz^{229} & \bz^{32}\\
  \bz^{253} & \bz^{205} & \bz^{204}\\
  \bz^{205} & \bz^{253} & \bz^{17}\\
  \bz^{254} & \bz^{230} & \bz^{228}\\
  \bz^{230} & \bz^{254} & \bz^{151}
  \end{matrix}
\end{pmatrix}\hspace{-3pt},
\]
%\end{footnotesize}
where $\bz=2$ is a primitive element of $\fte$ satisfying $\bz^8+\bz^4+\bz^3+\bz^2+\bz^0=0$ (see the list of Conway polynomials~\cite{ConwayPolynomials}). 
Note that we are using both the integer representation and the exponent representation of a finite field element. 
For instance, $122$ has the binary representation $01111010$,
corresponding to $\bz^6+\bz^5+\bz^4+\bz^3+\bz = \bz^{229}$.
%     1     .     .     .     .     . z^229  z^57 z^252
%     .     1     .     .     .     .  z^57 z^229  z^32
%     .     .     1     .     .     . z^253 z^205 z^204
%     .     .     .     1     .     . z^205 z^253  z^17
%     .     .     .     .     1     . z^254 z^230 z^228
%     .     .     .     .     .     1 z^230 z^254 z^151
This is a $\text{GRS}(A,6,\vlam)$~with $A=[0,8] = \{0,1,\bz,\bz^{25}=1~+~\bz,\bz^2,\bz^{50}=\bz^2+\bz^0,\bz^{26}=\bz^2+\bz,\bz^{198}=\bz^2+\bz~+~1,\bz^{3}\}$~and $\vlam = (\bz^{177}, \bz^{177}, \bz^5, \bz^5, \bz^{234}, \bz^{234}, \bz^{208}, \bz^{208}, \bz^{119})$. 
Using \cite[Ch.~10, Thm.~4]{MW_S}, we can deduce that the dual of this $\text{GRS}(A,6,\vlam)$ is an $\text{GRS}(A,3,\vga)$, where $\vga = (\bz^{47}, \bz^{82}, \bz^{171}, \bz^{239}, \bz^{221}, \bz^{144}, \bz^{75}, \bz^{199}, 1)$. We believe that the selection of $A$ as the first $n$ nonegative integers in ISA-L is purely for the convenience of the for-loops in its C code and has no significant reasons behind. 

\subsection{Code Construction in Backblaze Vaults}

Backblaze Vaults (\cite{BackblazeVaultsBlog, BackblazeVaultsGithub}) uses a systematic generator matrix $\bG = [\bI_k \mid \bV_1^{-1}\bV_2] = \bV_1^{-1}\bV$, where $\bV=[\bV_1\mid \bV_2]$ is the $k\times n$ Vandermonde matrix created by the finite field elements corresponding to $0,1,\ldots,n-1$. In other words, this is a standard $\rsk$ in which $A = [0,n-1]$, the same set of evaluation points as in the $\grskl$ of ISA-L.

\subsection{Code Construction in Facebook's f4}
The GRS(14,10) used in Facebook's f4 storage system~\cite{FBf42014} has been used repeatedly in the literature to demonstrate either the inefficiency of Reed-Solomon code's naive repair scheme or improvements from it. At the moment we could no longer locate the official source of the implementation of this code. However, from the previous studies~\cite{Shanmugam2014, GuruswamiWootters2016, GuruswamiWootters2017, DuursmaDau2017} and from our own copy of the code, GRS(14,10)~in~f4 was constructed via the generator polynomial (see~\cite[Chp.~7]{MW_S}) $g(x)=(x-1)(x-\bz)(x-\bz^2)(x-\bz^3)$. 
In general, the generator polynomial is $g(x)=\prod_{i=0}^{r-1} (x-\bz^i)$ and the codewords correspond to vectors of coefficients of all polynomials $c(x)=\sum_{i=0}^{n-1}\bc_i x^i\in \fte[x]$ that admit $1,\bz,\bz^2,\ldots,\bz^{r-1}$ as roots, where $r = n - k$. Equivalently, the code has a parity check matrix $\bH = \text{Vand}(1,\bz,\bz^2,\ldots,\bz^{n-1})$, and hence, it is the dual of an $\rsk$ with $A = \{1,\bz,\bz^2,\ldots,\bz^{n-1}\}$.

%\newpage
To summarize, all (valid) constructions of (generalized) Reed-Solomon codes in industry available to us, although may look different, are still the same as the original one provided in Section~\ref{subsec:classic}.
This is good news because we can focus on %repairing %(generalized) Reed-Solomon codes 
%them using 
just the original construction. Although straightforward, this section serves as a one-stop reference for future studies in this direction. \vspace{-10pt}

\section{Heuristic Search for Low-Bandwidth Schemes}
\label{sec:algorithms}

\subsection{Trace Repair Framework}
\label{subsec:tracerepair}

Following the framework developed in~\cite{GuruswamiWootters2016, GuruswamiWootters2017}, a (linear) trace repair scheme for the component $\bcjs$ of a codeword $\vc$ of a linear $[n,k]$ code $\C$ over $\fql$ corresponds to a set of $\ell$ \textit{dual} codewords $\left\{\bgi\right\}_{i\in [\ell]} \subset \Cd$, $\bgi=\left(\bgi_1,\ldots,\bgi_n\right)$, satisfying the \textit{Full-Rank Condition}:
$
\rank\big\{\bgijs\colon i\in [\ell]\big\}=\ell$. 
Such a repair scheme is denoted $\Rg$, which can be viewed as an $\ell\times n$ \textit{repair matrix} with $\bgij$ as its $(i,j)\text{-entry}$. As established in~\cite{GuruswamiWootters2016, GuruswamiWootters2017}, the repair bandwidth of such a repair scheme (in \textit{bits}) is $b(\R) = \sum_{j\in[n]\setminus \{j^*\}} r_j$, where $r_j \triangleq \rank\left(\left\{\bgij\colon i \in [\ell]\right\}\right)$.
To repair all $n$ components of~$\vc$, we need $n$ such repair schemes (possibly with repetition). 
See, e.g., \cite{DauDinhKiahTranMilenkovic2021}, for a detailed explanation of why the above scheme works with an example.
We describe an implementation of this repair scheme later in Section~\ref{sec:implementation}. 

Note that as the dual of a (generalized) Reed-Solomon code is another generalized Reed-Solomon code, searching for a set of dual codewords is equivalent to searching for a set of polynomials $\left\{g_i(x)\colon i \in [\ell]\right\}\subset \fql[x]$ of degree at most $r-1$. Using the notation above, we have $\bgij = \bla_j g_i(\balj)$. 
As $\bla_j$'s do not affect the repair bandwidth, we usually ignore them.

\subsection{Heuristic Search for Low-Bandwidth Repair Schemes}
\label{subsec:algorithms}

As discussed in Section~\ref{sec:constructions} and Section~\ref{subsec:tracerepair}, to construct low-bandwidth repair schemes for Reed-Solomon codes over $\fte$, one must find sets of eight check polynomials $\left\{g_i(x)\right\}_{i\in[8]}\subset \fte[x]$ that satisfy the Full-Rank Condition while incurring low repair bandwidths.

There are two main types of check polynomials applicable for Reed-Solomon codes over $\fte$: the \textit{algebraic} ones are based on algebraic structures such as subfields and subspaces~\cite{GuruswamiWootters2016, GuruswamiWootters2017, DauMilenkovic2017, DauDinhKiahTranMilenkovic2021, BermanBuzagloDorShanyTamo_ISIT_2021, LiWangJafarkhani_CommLett_2021} or cosets of subfields~\cite{LiWangJafarkhani-Allerton-2017, LiWangJafarkhani_TIT_2019}, while the others are found by a \textit{computer search}~\cite{GuruswamiWootters2016, DuursmaDau2017}. In some special cases, for example when $r = 4$ and $n\geq 12$ (e.g., GRS(14,10) or GRS(12,8)), the algebraic constructions generate the lowest known repair bandwidths~\cite{LiWangJafarkhani_CommLett_2021}. But for many other cases, an algebraic construction only yields an insignificant reduction from the naive bandwidth, e.g., $6\%$ when $r = 3$ and $n = 11$, whereas a computer search can produce a scheme achieving a much higher reduction, e.g., $28\%$ in this case (see Table~\ref{tab:r=3}).

Note that there are many different RS$(n,k)$'s depending on which set of evaluation points $A$ is chosen (we will show later that different $A$'s may lead to different (optimal) repair bandwidths). 
In this work, we examine two types of codes: 
\begin{itemize}
	\item ISAL-codes: $A=[0,n-1]\subseteq\fte$, $n\leq 256$,
	\item $\fst$-based codes: $A\hspace{-2pt}=\hspace{-2pt}\{0,1,\bz_{16},\ldots,\bz_{16}^{n-2}\}\hspace{-2pt}\subseteq \hspace{-2pt}\fst$, $n\leq 16$, where $\bz_{16}$ is a primitive element of $\fst$ satisfying $\bz_{16}^4+\bz_{16}+1=0$ (see the list of Conway polynomials~\cite{ConwayPolynomials}).
\end{itemize}
Searching for good repair schemes for ISAL-codes is much harder because of the very large search space (over $\fte$). For $\fst$-based codes, the search complexity is lower, which allows us to locate good schemes within a reasonable amount of time. The ``lifting" technique, which previously has been employed only in the context of algebraic constructions~\cite{TamoYeBarg2017,TamoYeBarg2018,LiWangJafarkhani-Allerton-2017, LiWangJafarkhani_TIT_2019, LiWangJafarkhani_CommLett_2021}, is now used in our heuristic algorithms to transform a repair scheme for codes over $\fst$ to a repair scheme for codes over $\fte$ (Proposition~\ref{pro:lifting}).   
We observe that the current-best repair bandwidths of $\fst$-based codes are always at least as low as those of ISAL-codes and in many cases are smaller.
Showing that this is true in general for Reed-Solomon codes over $\fte$ (or finding a counterexample) is an interesting open problem.

To find the lowest-bandwidth repair scheme, in general, we need to examine $\binom{P}{\ell}$ different sets of $\ell$ polynomials each, where $P$ is the number of candidate check polynomials. This number is huge even for very modest parameters. To reduce the search complexity, one needs to reduce $P$ and/or $\ell$. We discuss different ways to achieve complexity reduction below.

To reduce $P$, the number of candidate check polynomials, we make the following simplifying assumptions (see also~\cite{DuursmaDau2017}).
\begin{itemize}
	\item (A1) $\deg(g_i)=r-1$: we prove in Proposition~\ref{pro:highest_degree} that this assumption does \textit{not} lead to suboptimal solutions, and
	\item (A2) $g_i(x)$ has $r-1$ (possibly repeated) roots in $A$: this is based on the (unproven) intuition that more zeros in the repair matrix $\R$ may lead to a low repair bandwidth. This has been confirmed empirically in our various experiments.
\end{itemize}
\vspace{-3pt}

\begin{proposition}
\label{pro:highest_degree}
Let $\left\{g_i(x)\right\}_{i \in [\ell]}$ be a set of check polynomials of degree \textit{at most} $r-1$ corresponding to a repair scheme for $\bcjs$ of an $\rsk$ over $\fql$. Then there exists another set of check polynomials $\left\{h_i(x)\right\}_{i \in [\ell]}$ of degree \textit{exactly} $r-1$ that can repair $\bcjs$ with the same or smaller bandwidth.
\end{proposition}
\begin{proof}
\textbf{Case 1.} If one polynomial has degree exactly $r-1$, e.g., $\deg(g_1)=r-1$, then we set\vspace{-3pt}
\[
h_i(x)=
\begin{cases}
g_i(x),&\text{ if } \deg(g_i) = r-1,\\
g_i(x) + g_1(x), & \text{ if } \deg(g_i) < r-1.
\end{cases}\vspace{-3pt}
\]
Then $\deg(h_i)=r-1$ for every $i$ and moreover,\vspace{-3pt} 
\[
\spn\left(\{h_i(\bal)\}_{i \in [\ell]}\right) = \spn\left(\{g_i(\bal)\}_{i \in [\ell]}\right),
\]
for every $\bal \in A$. Thus, $\lhr$ is another repair scheme for $\bcjs$ and has the same bandwidth as $\lgr$.

\textbf{Case 2.} If $\deg(g_i)<r-1$ for every $i \in [\ell]$ then we select an $\balb \in A \setminus \{\baljs\}$ and set $h_i(x)=g_i(x)(x-\balb)^{r-1-\max_t\{\deg(g_t)\}}$. 
Then $\left\{h_i(x)\right\}_{i \in [\ell]}$ is another repair scheme for $\bcjs$ and has the same or smaller  bandwidth. Moreover, at least one $h_i$ has degree exactly $r-1$, which reduces this case to Case 1. 
\end{proof}
\vspace{-3pt}

To reduce $\ell$, the number of polynomials needed in a search, we use the well-known lifting and extension techniques. % for repair schemes of Reed-Solomon codes. 
The lifting technique allows us, for instance, to transform a repair scheme for a Reed-Solomon code constructed in $\fst$ to a repair scheme for another Reed-Solomon code constructed in $\fte$ using the same set of evaluation points while doubling the bandwidth.
The extension technique allows us, for example, to transform a repair scheme of a Reed-Solomon code over $\fte$ with the base field $\Ff$ into a repair scheme with base field $\ft$.
For completeness, we formalize these techniques below.\vspace{-3pt}
%To keep it general, let $\fq$ be a finite field of $q$ elements ($q$ is a prime power) and $\fql$ its extension field.

\begin{proposition}[Lifting]
\label{pro:lifting}
Suppose that $m\mid \ell$ and $\{g_i(x)\}_{i\in [m]}\subset \fqm[x]$ corresponds to a repair scheme with bandwidth $b$ (measured in elements in $\fq$) for the $j^*$-th component of an $\rsk$ over $\fqm$. Then $\{\bbj g_i(x)\}_{i\in [m],j \in [\ell/m]}$, where $\{\bbj\}_{j \in [\ell/m]}$ is an $\fqm$-basis of $\fql$, corresponds to a repair scheme with bandwidth $\frac{\ell}{m}b$ for the $j^*$-component of the $\rsk$ with the same evaluations points $A$ but constructed over~$\fql$. 
\end{proposition} \vspace{-10pt}

\begin{proposition}[Extension]
\label{pro:extension}
Suppose that $m$ divides $\ell$ and the set $\{g_i(x)\}_{i\in [\ell/m]}\subset \fql[x]$ corresponds to a repair scheme with bandwidth $b$ (measured in elements in $\fqm$) for the component $\bcjs$ of an $\rsk$ over $\fql$, treating $\fqm$ as the base field. Then the set $\{\bgaj g_i(x)\}_{i\in [\ell/m],j \in [m]}$, where $\{\bgaj\}_{j \in [m]}$ is an $\fq$-basis of $\fqm$, corresponds to a repair scheme with bandwidth $bm$ (measured in elements in $\fq$) for the $\bcjs$ of the same code. 
\end{proposition}%\vspace{20pt}

%\vspace{-10pt}
Applying the above complexity reduction assumptions and techniques, we  construct low-bandwidth repair schemes for ISAL-codes and $\fst$-codes utilizing the following algorithms.
\begin{itemize}
	\item \textbf{Algorithm 1}: (degree-four repair) Introduced in~\cite{DuursmaDau2017} to tackle GRS(14,10) over $\fte$, this algorithm first constructs a list of pairs of polynomials in $\fql[x]$ (treating $\ff_{q^{\ell/2}}$ as the base field) that has bandwidth at most a threshold $\theta_2$, and then search for sets of four polynomials (treating $\ff_{q^{\ell/4}}$ as the base field) consisting of two pairs from that list (keeping the first pair unchanged while adding a multiplicative factor to the second) that has bandwidth at most $\theta_4$ $(\theta_4 < \theta_2)$. Various thresholds $\theta_2$ and $\theta_4$ were tested to produce the lowest bandwidth.  
	\item \textbf{Algorithm 2}: (exhaustive search) For $r = 2, 3$, we can also apply a direct exhaustive search for sets of polynomials with low-bandwidths. Note that we do not have to wait for the algorithm to finish (which would take too long). We can retrieve the currently found bandwidths for all codeword components and stop if find them satisfactory or see that there is little chance to improve further. 
\end{itemize}
Note that both algorithms go through each valid set of polynomials \textit{once} and check which codeword components could be repaired by the set. Algorithm 1 terminates if repair schemes of bandwidth not exceeding the specified threshold have been found for \textit{all} codeword components. Best found bandwidths for codes with $4\hspace{-2pt}\leq\hspace{-2pt} n\hspace{-2pt}\leq\hspace{-2pt} 16$ and $2\hspace{-2pt} \leq\hspace{-2pt} r\hspace{-2pt} \leq\hspace{-2pt} 4$ are reported in Table~\ref{tab:r=2}, Table~\ref{tab:r=3}, and Table~\ref{tab:r=4}. Column ``$\fst$-based algebraic'' refers to repair schemes using subspace polynomials and lifting~\cite{LiWangJafarkhani_TIT_2019,DauDinhKiahTranMilenkovic2021}.
$\fst$-based heuristic always finds the best bandwidths, which could also be due to the fact that it is cheaper to search over $\fst$ than $\fte$. 
We maintain a web page~\cite{webpage} to keep track of the best bandwidths and the corresponding repair schemes.  

\begin{table}
\centering
\begin{tabular} {|>{\centering\arraybackslash}p{1.4cm}|>{\centering\arraybackslash}p{0.9cm}|>{\centering\arraybackslash}p{1.5cm}|>{\centering\arraybackslash}p{1.5cm}|>{\centering\arraybackslash}p{1.5cm}|}
\hline

\textbf{Redundancy} $r=2$ & \textbf{Default} & \textbf{ISA-L heuristic} & $\mathbb{F}_{16}$-\textbf{based algebraic}&  $\mathbb{F}_{16}$-\textbf{based heuristic}\\

\hline
$n=4$ & 16 & 12 (-25\%) & 18 (+12.5\%)& 12 (-25\%) \\

\hline
$n=5$ & 24 & 18 (-25\%)& 24 (-0\%)& 18 (-25\%) \\

\hline
$n=6$ & 32 & 24 (-25\%) & 30 (-6.3\%)& 24 (-25\%) \\

\hline
$n=7$ & 40 & 32 (-20\%) & 36 (-10\%)& 30 (-25\%) \\

\hline
$n=8$ & 48 & 38 (-20.8\%) & 42 (-12.5\%)& 38 (-20.8\%) \\

\hline
$n=9$ & 56 & 44 (-21.4\%) & 48 (-14.3\%) & 44 (-21.4\%) \\

\hline
$n=10$ & 64  &	50(-21.9\%) & 54 (-15.6\%)& 50 (-21.9\%) \\

\hline
$n=11$ & 72	& 58 (-19.4\%) & 60 (-16.7\%)& 56 (-22.2\%) \\

\hline
$n=12$ & 80	& 64 (-20\%) & 66 (-17.5\%)& 64 (-20\%) \\

\hline
$n=13$ & 88	& 72 (-18.2\%) & 72 (-18.2\%)& 70 (-20.5\%) \\

\hline
$n=14$ & 96 & 80 (-16.7\%) & 78 (-18.8\%)& 76 (-20.8\%) \\

\hline
$n=15$ & 104 & 84 (-19.2\%) & 84 (-19.2\%)& 84 (-19.2\%) \\

\hline
$n=16$ & 112 & 90 (-19.6\%) & 90 (-19.6\%)& 90 (-19.6\%) \\

\hline
\end{tabular}
\caption{Current-best repair bandwidths (measured in bits) for Reed-Solomon codes with $4 \leq n \leq 16$ and $r = 2$.}
\label{tab:r=2}
\vspace{-10pt}
\end{table}

\begin{table}
\centering
\begin{tabular} {|>{\centering\arraybackslash}p{1.4cm}|>{\centering\arraybackslash}p{0.9cm}|>{\centering\arraybackslash}p{1.5cm}|>{\centering\arraybackslash}p{1.5cm}|>{\centering\arraybackslash}p{1.5cm}|}
\hline

\textbf{Redundancy} $r=3$ & \textbf{Default} & \textbf{ISA-L heuristic} & $\mathbb{F}_{16}$-\textbf{based algebraic}&  $\mathbb{F}_{16}$-\textbf{based heuristic}\\

\hline
$n=4$ & 8 & 8 (-0\%) & 18(+100\%)& 8 (-0\%) \\

\hline
$n=5$ & 16 & 12 (-25\%)	 & 24 (+50\%)& 12 (-25\%) \\

\hline
$n=6$ & 24	& 16 (-33.3\%) & 30 (+25\%)& 16 (-33.3\%) \\

\hline
$n=7$ & 32 & 22 (-31.3\%) & 36 (+12.5\%)& 22 (-31.3\%) \\

\hline
$n=8$ & 40	& 28 (-30\%) & 42 (+5\%)& 28 (-30\%) \\

\hline
$n=9$ & 48 & 34 (-29.2\%) & 48 (0\%)& 32 (-33.3\%) \\

\hline
$n=10$ & 56	& 40 (-28,6\%) & 54 (-3.6\%)& 40 (-28,6\%) \\

\hline
$n=11$ & 64	& 46 (-28.1\%) & 60 (-6.3\%)& 46 (-28.1\%) \\

\hline
$n=12$ & 72	& 52 (-27.8\%) & 66 (-8.3\%)& 52 (-27.8\%) \\

\hline
$n=13$ & 80	& 58 (-27.5\%) & 72 (-10\%)& 58 (-27.5\%) \\

\hline
$n=14$ & 88 & 66 (-25\%) & 78 (-11.4\%)& 64 (-27.3\%) \\

\hline
$n=15$ & 96	& 72 (-25\%) & 84(-12.5\%)& 70 (-27.1\%) \\

\hline
$n=16$ & 104 & 78 (-25\%) & 90 (-13.5\%)& 76 (-26.9\%) \\

\hline
\end{tabular}

\caption{Current-best repair bandwidths (measured in bits) for Reed-Solomon codes with $4 \leq n \leq 16$ and $r = 3$.}
\label{tab:r=3}
\end{table}

\begin{table}
\centering
\begin{tabular} {|>{\centering\arraybackslash}p{1.4cm}|>{\centering\arraybackslash}p{0.9cm}|>{\centering\arraybackslash}p{1.5cm}|>{\centering\arraybackslash}p{1.5cm}|>{\centering\arraybackslash}p{1.5cm}|}
\hline

\textbf{Redundancy} $r=4$ & \textbf{Default} & \textbf{ISA-L heuristic} & $\mathbb{F}_{16}$-\textbf{based algebraic}&  $\mathbb{F}_{16}$-\textbf{based heuristic}\\

\hline
$n=5$ & 8 & 8 (-0\%)& 16 (+100\%)& 8 (-0\%) \\

\hline
$n=6$ &	16 & 12 (-25\%)& 20 (+25\%)& 12 (-25\%) \\

\hline
$n=7$ & 24 & 16 (-33.3\%)& 24 (0\%)& 16 (-33.3\%) \\

\hline
$n=8$ & 32 & 22 (-31.3\%)& 28 (-12.5\%)& 	22 (-31.3\%) \\

\hline
$n=9$ & 40	& 28 (-30\%)& 32 (-20\%)& 26 (-35\%) \\

\hline
$n=10$ & 48	& 36 (-25\%)& 36 (-25\%)& 32 (-33.3\%) \\

\hline
$n=11$ & 56	& 42 (-25\%)& 40 (-28.6\%)& 38 (-32.1\%) \\

\hline
$n=12$ & 64	& 48 (-25\%)& 44 (-31.3\%)& 44 (-31.3\%) \\

\hline
$n=13$ & 72	& 54 (-25\%)& 48 (-33.3\%)& 48 (-33.3\%) \\

\hline
$n=14$ & 80	& 62 (-22.5\%)& 52 (-35\%)& 52 (-35\%) \\

\hline
$n=15$ & 88	& 68 (-22.7\%)& 56 (-36.4\%)& 56 (-36.4\%) \\

\hline
$n=16$ & 96	& 60 (-37.5\%)& 60 (-37.5\%)& 60 (-37.5\%) \\

\hline
\end{tabular}
\caption{Current-best repair bandwidths (measured in bits) for Reed-Solomon codes with $5 \leq n \leq 16$ and $r = 4$.}
\label{tab:r=4}
\vspace{-10pt}
\end{table}

\begin{definition}[Bandwidth profile]
Given $n>0$ and $\ell > 0$, a \textit{bandwidth profile} $\vb=(b_1,b_2,\ldots,b_n)$, $b_j \in [\ell]$, is \textit{feasible} for an $\rsk$, where $A \subseteq \fql$, $|A|=n$, if there exists a collection of $n$ repair schemes that require bandwidth $b_j$ for the $j$-th components of its codeword, $j\in [n]$. 
A bandwidth profile~$\vbs$ is \textit{optimal} if for every $j \in [\ell]$, $b^*_j$ is the lowest bandwidth possible to (linearly) repair the $j$-th codeword component of the code.  
\end{definition}

Note that in Table~\ref{tab:r=2}, Table~\ref{tab:r=3}, and Table~\ref{tab:r=4}, we report $\mvb \triangleq \max\{b_j\colon j \in [\ell]\}$, where $\vb$ is a feasible bandwidth profile.  
Individual components may require lower bandwidths.
Visit our web page~\cite{webpage} for the most updated bandwidths.

\subsection{The Impact of Evaluation Points}
\label{subsec:evalPoints}

In this section we discuss important issues regarding the selection of the evaluation points in Reed-Solomon codes. In particular, we demonstrate via an example that different sets of evaluation points may lead to different repair bandwidths.

\begin{proposition}
\label{pro:A}
Translating and dilating the set of evaluation points of a Reed-Solomon code do not affect its optimal bandwidth profile. In other words, $\text{RS}(A,k)$ and $\text{RS}(\bbe A+\bga,k)$, where $\bbe, \bga \in \fql$, $\bbe \neq 0$, have the same optimal bandwidth profile (up to a permutation).
\end{proposition}  
\begin{proof}
Suppose that $\{g_i(x)\}_{i\in [\ell]}$ is a repair scheme for $\rsk$, which can be used to repair $f(\balj)$, $\balj \in A$, and has bandwidth~$b$. We set $h_i(x) \triangleq g_i\left((x-\bga)/\bbe\right)$, $i\in [\ell]$. Then $h_i(\bbe\bal+\bga) = g_i(\bal)$ for every $\bal \in A$. Therefore, $\{h_i(x)\}_{i\in [\ell]}$ can be used to repair $\bbe\balj+\bga\in \bbe A+\bga$ with bandwidth~$b$.
Hence, $\bbe\bal+\bga\in \bbe A+\bga$ can be repaired in $\text{RS}(\bbe A+\bga,k)$ with a bandwidth not exceeding that for $\bal\in A$ in $\rsk$.
Since $A = \bbe^{-1}(\bbe A+\bga)-\bbe^{-1}\bga$, the same argument proves that the reverse conclusion is also true. Thus, these two codes have the same optimal repair bandwidth for their corresponding codeword components ($f(\bal)$ and $f(\bbe\bal+\bga)$).
\end{proof} 

As a corollary of Proposition~\ref{pro:A}, although there are $\binom{q^\ell}{n}$ different RS$(n,k)$ codes over $\fql$, we can divide them into classes of codes that have evaluation points obtained from each other by translations and dilations. Each class can have at most $q^\ell(q^\ell-1)$ members with the same optimal bandwidth profile (up to permutations). Identifying other transformations of $A$ that preserve the optimal bandwidth profile is an open problem.
 
%Heuristic searches for low-bandwidth repair schemes for GRS(14,1) used in Facebook's f4 were conducted in~\cite{GuruswamiWootters2016, DuursmaDau2017}. 

To examine the impact of evaluation points on repair bandwidths, we consider RS(5,3) codes over $\fst$, which have small $r$ and field size and hence allow us to determine their optimal bandwidth profiles. An RS(5,3) (with a given generator matrix) was first investigated in~\cite{Shanmugam2014}.
%An RS(5,3) was first investigated in~\cite{Shanmugam2014}, which has the following generator matrix.
\begin{comment}
\vspace{-5pt}
\begingroup
\begin{small}
\renewcommand*{\arraystretch}{1.5}
% your pmatrix expression
\[
\bG = 
\begin{pmatrix}
\bI_3 & \rvline & 
  \begin{matrix}
  \frac{(\bom^4-\bom^2)(\bom^4-\bom^3)}{(\bom-\bom^2)(\bom-\bom^3)} & \frac{(\bom^5-\bom^2)(\bom^5-\bom^3)}{(\bom-\bom^2)(\bom-\bom^3)} \\
  \frac{(\bom^4-\bom)(\bom^4-\bom^3)}{(\bom^2-\bom)(\bom^2-\bom^3)} & \frac{(\bom^5-\bom)(\bom^5-\bom^3)}{(\bom^2-\bom)(\bom^2-\bom^3)} \\
  \frac{(\bom^4-\bom)(\bom^4-\bom^2)}{(\bom^3-\bom)(\bom^3-\bom^2)} & \frac{(\bom^5-\bom)(\bom^5-\bom^2)}{(\bom^3-\bom)(\bom^3-\bom^2)}
  \end{matrix}
\end{pmatrix},\vspace{-5pt}
\]
\end{small}
\endgroup 
\end{comment}
%where $\bom=\bz_{16}^3$ is a fifth root of unity, and $\bz_{16}$ is a primitive element of $\fst$ satisfying $\bz_{16}^4+\bz_{16}+1=0$.
We searched through all 524,160 different arrangements of five elements in $\fst$ and found 240 different $A$'s (orders of elements are important) giving rise to the same code, e.g., $A =\{0, 1, \bz_{16}^{14}, \bz_{16}^8, \bz_{16}^{12}\}$,  
where 
%$\bom=\bz_{16}^3$ is a fifth root of unity, and 
$\bz_{16}$ is a primitive element of $\fst$ satisfying $\bz_{16}^4+\bz_{16}+1=0$.
It was shown in~\cite{Shanmugam2014} that the optimal repair bandwidth for each systematic component $(j=1,2,3)$ is 10 bits.
On the other hand, from Table~\ref{tab:r=2} and~\cite{webpage}, we know that 9 bits are sufficient to repair other RS(5,3) codes.
This shows that different sets of evaluation points can lead to different repair bandwidths.   

In fact, we have obtained a complete picture of the bandwidth profiles of \textit{all} RS$(n,n-2)$ codes, $4 \leq n \leq 16$, over $\fst$.
First, from the list of $16$ \textit{monic} polynomials of degree one over $\fst$, we created a list of $6,142,500=\binom{16}{4}15^3$ sets of four check polynomials (keeping the first monic while adding arbitrary nonzero coefficients to others).
We then generate a \textit{rank profile} $\vec{{\sf r}} = (r_1,r_2,\ldots,r_{16})$ for each polynomials set and remove those whose components are all smaller than 4, which took our GAP program 10 hours to complete. More than six millions sets of polynomials remain as potential repair schemes.
For each $A \subseteq \fst$, we determine the optimal repair bandwidth for each codeword component of $\rsk$ by going through all the rank profiles, restricting to the positions in $A$. For instance, when $n = 5$ and $k = 3$, among $4368=\binom{16}{5}$ different 5-subsets $A\subset\fst$, 2880 have bandwidth profile $(9,9,9,9,8)$, 1440 have bandwidth profile $(9,9,9,9,9)$, while only 48 have bandwidth profile $(10,10,10,10,10)$. It turns out that the RS(5,3) examined in~\cite{Shanmugam2014} is accidentally among the minority that require 10 bits. Most others need only 9 bits.

\section{An Implementation of Trace Repair Schemes in C}
\label{sec:implementation}

We implemented a trace repair scheme~\cite{implementation} on top of the existing C code in the ISA-L~\cite{ISAL}. As the bandwidth reduction is known, we focus on minimizing the computational complexity. Following our notation in the C code, Node $j$ aims to recovers $\bcj$ by downloading relevant data from Nodes $i$, $i\in I \subseteq [n]\setminus \{j\}$. In ISA-L $|I|=k$ while in our implementation $|I|=n-1$. The code implemented in ISA-L is systematic, i.e., the first $k$ components are data and the last $n-k$ components of each codeword are parities.  
Finite field elements over $\fte$ are represented as integers in $[0,255]$ (see Section~\ref{sec:pre}). %Additions are done by performing bitwise XOR operation while multiplications are based on table lookups. 

\subsection{ISA-L Implementation} 
ISA-L uses the naive scheme to repair the generalized Reed-Solomon codes, which means that each helper Node $i$ simply reads and sends $\bci$ without performing any computation. Node~$j$, after receiving data from $k$ nodes, can recover $\bcj$ as follows. A lost parity component can be repaired by downloading the data components (systematic part) and performing encoding. A missing data component, however, requires more work: Node $j$ performs first a matrix inversion to obtain $(\bG[I])^{-1}$, where $\bG[I]$ denotes the submatrix of $\bG$ consisting of columns of $\bG$ indexed by $I$, and then multiply this matrix with the vector consisting of $\bci, i\in I$. This procedure is repeated for $T$ different codewords, where $T$ is a large number representing the number of codewords to be repaired. %Roughly, $T=F/k$, where $F$ is the file size in bytes. 
For instance, if the encoded file has size 60 MB, and an RS(9,6) is used, then $T$ is roughly ten million.
Note that the matrix inversion is done only once. 
%The input to the repair scheme is represented as a 2-dimensional array/matrix $\bc[n][T]$ in which each column corresponds to a codeword.

\subsection{Our Implementation}
In a trace repair scheme, as opposed to ISA-L, both senders and receiver perform computations. As a large number of codewords are being repaired, say, millions, it is crucial to identify parts of the computation that could be precomputed and stored for fast access. We convert the $n$ repair schemes (for $n$ components) into three lookup tables: $H$ (helper) allows the helper nodes to create the repair traces (bits) while $R$ (recover) and $D$ (dual basis) allow the receiver node to process the repair traces and recover the lost component. Please refer to~\cite{DauDinhKiahTranMilenkovic2021} for undefined terminologies.
A frequently used operation is the XOR-sum of the bits of an integer.
Hence, we precompute an array called $\sp$ that store these values for all $m\in [0,255]$.
We describe below the computation for one codeword. 
%All of these can be done once and offline before the repair begins. 
%ISA-L stores the $T$ codewords in a 2-dimensional array $c[n][T]$ where each column is a codeword of length $n$.

%Continue the discussion on trace repair schemes in Section~\ref{subsec:tracerepair}, taking an RS(5,3) given a set of $\ell$ dual codeword $\Rg$

\textbf{Sender side.} Node $i$ extracts the number of traces to be sent $r_i = H[i][j][0]$, and uses $r_i$ numbers $H[i][j][1]$ to $H[i][j][r_i]$ to compute $r_i$ repair traces. Table $H$ is defined in a way that the $s$-th trace from Node $i$ is the inner product of $H[i][j][s]$ and $\bci$, %\vspace{-3pt} 
\[
\srt_i[s]=\sp\big[H[i][j][s]\ \&\ \bci\big],\quad s \in [1,r_i].\vspace{-3pt}
\] 
Node~$i$ then sends $\srt_i$ to Node~$j$. 

\textbf{Receiver side.} Node $j$ uses $\srt_i$, $i \in [n]\setminus \{j\}$, $R$, and $D$ to recover $\bcj$ as follows. For each $i$, it generates eight column traces $\sct_i[s]$, $s\in [8]$, using the formula\vspace{-3pt}
\[
\sct_i[s] = \sp\big[R[i][j][s]\ \& \ {\sf{Dec}}(\srt_i)\big],\vspace{-3pt}
\]
where ${\sf{Dec}}(\srt_i)$ turns the $r_i$ bits in $\srt_i$ into a decimal number, ready for bitwise operations. ${\sf{Dec}()}$ is also implemented using bit-shift and XOR operations. %Still, computing ${\sf{Dec}}(\srt_i)$ is the most time consuming step (70\% total computation time). 
Finally,\vspace{-3pt}
\[
\bcj = \oplus_{s=1}^8 \big( \big(\oplus_{i \in [n]\setminus \{j\}}\sct_i[s]\big)\times D[s]\big).%\vspace{-3pt}
\] 

%\vspace{-1pt}
\begin{table}[t]
\centering
\begin{tabular}{|l|c|c|c|}
\hline
		& Senders 	& Receiver 	& Total\\
\hline
ISA-L (naive) 	& 0	(sec)		& 0.57 (sec) 		& 0.57 (sec)\\
\hline
Trace repair &0.2 (sec)	&1.4 (sec)		& 1.6 (sec)\\
\hline
\end{tabular}
\caption{The running times of the naive repair (ISA-L) and trace repair for an RS(9,6) over \textit{ten millions} codewords and (random) single erasures (on a Linux server: Intel(R) Xeon(R) CPU E5-2690 v2 @ 3.00GHz with 792 GB RAM). For senders, the maximum running time among all helpers was used.}\vspace{-15pt}
\label{tab:comparison}
\end{table}

In our implementation~\cite{implementation}, we further optimize the code by joining the above steps at the receiver to save time.
%The key point is to move as many computation tasks as possible out of the for-loop with $T$ because $T$ is very large. 
%It turns out the bottle neck is the transformation of $T\times r_i$ bits from each Node $i$ into $T$ decimal numbers as the algorithm must include another loop over $r_i$ bits. 
For RS(9,6), our implementation is 2.8x slower than ISA-L (Table~\ref{tab:comparison}). For other codes such as RS(11,8), RS(16,13), RS(12,8), RS(14,10), RS(16,12), the gap is 1.8x-2.4x. 
Despite being a negative result, the small gaps are encouraging because trace repair is inherently more complicated.
Further optimizations may also reduce the gaps. 
As shown in \cite[Fig.~1]{MitraPantaRaBagchi2016}, the network transfer time is much larger than the computation time (say, 40 times) in naive repair. Hence, the reduction in bandwidth can well compensate the computation time and make trace repair faster. For example, a combination of 33\% reduction in bandwidth and 200\% increase in computation time could still lead to a 1.25x speed-up compared to naive repair. Implementing trace repairs on Hadoop 3 to verify this observation is a future work.
\vspace{-5pt}

\section*{Acknowledgement}

This work has been supported by the 210124 ARC DECRA Grant DE180100768. We thank Nguyen Dinh Quang Minh and Dau Trung Dung for their help in the implementation.% and Nguyen Phi Long for helping with the web page.
\newpage
\bibliographystyle{IEEEtran}
\bibliography{PracticalRS}
\end{document}